\documentclass[12pt, reqno]{amsart}
\usepackage{amssymb}
\usepackage{amsmath}
\usepackage{graphicx, color}
\usepackage{wrapfig,framed}

\usepackage[height=22cm, width=16.5cm, hmarginratio={1:1}]{geometry}
\usepackage{hyperref}

\theoremstyle{plain}
\newtheorem{theorem}{Theorem}

\newtheorem{lemma}{Lemma}

\theoremstyle{definition}

\newcommand{\beq}{\begin{equation}}
\newcommand{\eeq}{\end{equation}}

\newcommand{\CC}{\mathbb{C}}
\newcommand{\e}{\epsilon}
\newcommand{\p}{\partial}
\newcommand{\F}{\mathcal{F}}
\newcommand{\cH}{\mathcal{H}}
\newcommand{\bbT}{{\bf T}}
\newcommand{\bt}{{\bf t}}

\def\={\; = \;}
\def\+{\; + \;}
\def\:={\; := \; }

\begin{document}

\title[Okuyama--Sakai conjecture]{On a new proof of the Okuyama--Sakai conjecture}
\author{Di Yang, Qingsheng Zhang}

\begin{abstract}
In~\cite{OS} Okuyama and Sakai gave a conjectural equality for the higher genus 
generalized Br\'ezin--Gross--Witten (BGW) free  energies. 
In a recent work~\cite{YZ22} we established the Hodge-BGW correspondence on the relationship between 
certain special cubic Hodge integrals and the generalized BGW correlators, and 
a proof of the Okuyama--Sakai conjecture was also given {\it ibid}. In this paper, 
we give a new proof of the Okuyama--Sakai conjecture by a further application of the Dubrovin--Zhang theory for 
the KdV hierarchy.  
\end{abstract}

\maketitle
\tableofcontents

\section{Introduction}
The Br\'ezin--Gross--Witten (BGW) model was introduced in~\cite{BG,GW}:
\beq
Z_{\rm BGW}(A,A^{\dagger};\hbar)\sim \int\,[dU]\,e^{\frac{1}{\hbar^2}\,{\rm tr}(A^{\dagger}U\,+\,AU^{\dagger})}\, ,
\eeq
where the integration is over $M\times M$ unitary matrices with Haar measure $[dU]$ and $A$ is an $M\times M$ complex matrix.
In~\cite{MMS}, a one parameter deformation of $Z_{\rm BGW}$
was given via a generalized Kontsevich model~\cite{A, BR, KMMM,K}: 
\beq
Z_{\rm gBGW}(N,{\bf T};\hbar) \sim \int\, [d\Phi]\, e^{\frac{1}{\hbar^2}\, {\rm tr}\big(\Lambda^2\,\Phi\,+\,\frac{1}{\Phi}\,+\,(N\,-\,M)\,\hbar^2\,\log\,\Phi\big)}\, ,
\eeq
where ${\bf T}=(T_1,T_3,T_5,\dots)$,
$T_{2k+1}:= (2k-1)!!\, {\rm tr}\Lambda^{-2k-1}$, $k\ge0$, and $N$ is an indeterminate.
The normalization constant for $Z_{\rm gBGW}(N,{\bf T};\hbar)$ is often made so that
\beq\label{Zgbgw1}
Z_{\rm gBGW}(N, {\bf 0}; \hbar) \;\equiv\; 1 \,. 
\eeq
We call $Z_{\rm gBGW}(N,{\bf T};\hbar)$ the {\it normalized generalized BGW partition function with the parameter~$N$}.
The logarithm $\log Z_{\rm gBGW}(N, {\bf T}; \hbar) =: \F_{\rm gBGW}(N, {\bf T}; \hbar)$, belonging~\cite{A} to $\CC[\hbar][[N,{\bf T}]]$,  
is called the {\it normalized generalized BGW free energy with the parameter~$N$}.

Following Alexandrov~\cite{A},  introduce 
\beq
x \= N\, \hbar \,\sqrt{-2}\,.
\eeq
Then the free energy $\F_{\rm gBGW}(N, {\bf T}; \hbar)$ has the genus expansion~\cite{A}:
\beq\label{genusbgsalex}
\F_{\rm gBGW}\biggl(\frac{x}{\hbar\,\sqrt{-2}}, {\bf T}; \hbar\biggr) 
\;=:\; \sum_{g\geq 0}\hbar^{2g-2}\,\mathcal F^{\rm gBGW}_{g}(x, {\bf T}).
\eeq
We call $\F_{\rm gBGW}\bigl(\frac{x}{\hbar\,\sqrt{-2}}, {\bf T}; \hbar\bigr)$ the {\it normalized generalized BGW free energy}, 
$\mathcal F^{\rm gBGW}_{g}(x, {\bf T})$ its {\it genus~$g$ part}, and 
$Z_{\rm gBGW}\bigl(\frac{x}{\hbar\,\sqrt{-2}}, {\bf T}; \hbar\bigr)$ the {\it normalized generalized BGW partition function}.

According to~\cite{A} (see also~\cite{BR, GN, MMS}), the normalized generalized BGW partition function 
$Z_{\rm gBGW}\bigl(\frac{x}{\hbar\,\sqrt{-2}}, {\bf T}; \hbar\bigr)$ satisfies the Virasoro constraints, 
leading to the topological recursion of the Chekhov--Eynard--Orantin type \cite{A, CGG, DN, WYZ}
for the computation of the corresponding connected correlators.
It is also known that the normalized generalized BGW partition function 
is a particular tau-function for the Korteweg--de Vries (KdV) hierarchy (see e.g.~\cite{BR,MMS}). 
This enables one to apply theories of tau-functions for the KdV hierarchy 
to the study of $Z_{\rm gBGW}\bigl(\frac{x}{\hbar\,\sqrt{-2}}, {\bf T}; \hbar\bigr)$. 
Recall that the matrix-resolvent method~\cite{BDY, DYZ} gives the explicit formulae for the generating series of the logarithmic derivatives 
of an arbitrary KdV tau-function; using this method, explicit formulae for the generating series of the 
$n$-point generalized BGW correlators were 
obtained~\cite{DYZ} (see also~\cite{BR} for other proofs for the explicit formulae).
Recalling also that the KdV hierarchy is a reduction of the Kadomtsev--Petviashvili (KP) hierarchy, 
one can interprete $Z_{\rm gBGW}\bigl(\frac{x}{\hbar\,\sqrt{-2}}, {\bf T}; \hbar\bigr)$ as a point  
in the Sato Grassmannian for the KP hierarchy; in particular, the corresponding affine coordinates were calculated out~\cite{DYZ, Fu, Zhoudessin}.
The KdV hierarchy can also be viewed as a reduction of the BKP hierarchy~\cite{A2}, and 
the BKP affine coordinates for $Z_{\rm gBGW}\bigl(\frac{x}{\hbar\,\sqrt{-2}}, {\bf T}; \hbar\bigr)$ were given in~\cite{WY, WYZ}.
A new formula for $Z_{\rm gBGW}$ based on Virasoro constraints and the KdV/BKP theory  
was recently obtained in~\cite{A1, A0, LY}.

Another important theory of tau-functions for the KdV hierarchy was partially motivated from the 
quantum gravity and topological field theories~\cite{DVV, DW, Du1, W},  
and was systematically developed by Dubrovin and Zhang~\cite{DZ-norm} in the framework 
normal forms of evolutionary PDEs. In our previous work~\cite{YZ22}, we applied this theory from viewpoints of 
Virasoro constraints, and found the {\it Hodge-BGW correspondence} (for details about the 
Hodge-BGW correspondence see Section~\ref{section2} below). In particular, by using the Hodge-BGW correspondence and 
by deriving the loop equations 
we proved~\cite{YZ22} a conjecture of Okuyama and Sakai~\cite{OS}. 

\smallskip

\noindent {\bf The Okuyama--Sakai Conjecture}~(\cite{OS}).
{\it	Define a power series $y(x,{\bf T})\in\CC[[x+2]][[{\bf T}]]$ by 
\beq
y(x,{\bf T})\=\frac{\p^2 \F^{\rm gBGW}_0(x,{\bf T})}{\p T_1^2}. 
\eeq
For every $g\geq 1$, the genus g part of the generalized BGW free energy satisfy the identity: for $g=1$,
	\begin{align}
	& \F^{\rm gBGW}_{1}(x,{\bf T}) \= \frac1{24} \log \bigg(\frac{\p y(x,{\bf T})}{\p T_1}\bigg) \,-\, \frac{\log 2}{24} 
	\,-\, \frac1{12} \, \log\Bigl(-\frac{x}{2}\Bigr) \,, \label{eqn:bgwf1-wkjet}
	\end{align}
	and for $g\geq 2$, 
	\beq\label{eqn:kw-bgwjet}
	\F^{\rm gBGW}_{g}(x,{\bf T}) \= F^{\rm WK}_g\bigg(\frac{\p y(x,{\bf T})}{\p T_1},\dots,\frac{\p^{3g-2} y(x,{\bf T})}{\p T_1^{3g-2}}\bigg)\,
-\,\frac{1}{x^{2g-2}}\frac{(-1)^g \, 2^{g-1} \, B_{2g}}{2g \, (2g-2)}\,.
	\eeq
Here, $B_k$ denotes the $k$th Bernoulli number, and  $F^{\rm WK}_g(z_1,\dots,z_{3g-2})$, $g\ge2$, are 
certain specific functions of $(3g-2)$ variables (see \eqref{jstwk1-1}, \eqref{jstwk1-2}, \eqref{jstwk2} for the definitions). 
}

\smallskip
We recall that in~\cite{OS} the above conjectural identities \eqref{eqn:bgwf1-wkjet}, \eqref{eqn:kw-bgwjet} were verified for $g=1,2$, and were also checked~\cite{OS} for special evaluations up to $g=20$.

The goal of this paper is to
give a new proof of the Okuyama--Sakai conjecture 
by using the Hodge-BGW correspondence, and by considering a further application of the 
Dubrovin--Zhang theory~\cite{DZ-norm} (see also~\cite{Y23, YangZhou}) to the KdV hierarchy which can be  
interpreted like in~\cite{YangZhou} for the Laguerre Unitary Ensemble case (see also~\cite{Dubrovinuniversal}) as a universality class of criticality in the renormalization theory 
of quantum field theories.

\begin{theorem}\label{OSthm}
The Okuyama--Sakai conjecture holds. 
\end{theorem}

The rest of the paper is organized as follows. In Section~\ref{section2} we give a review of the Hodge-BGW correspondence. 
In Section~\ref{section3} we prove Theorem~\ref{OSthm}. 

\medskip

\noindent {\bf Acknowledgements.} 
The work was partially supported by 
the CAS Project for Young Scientists in Basic Research No.~YSBR-032 and by the NSFC No.~12061131014.

\section{Review of the Hodge-BGW correspondence}\label{section2}
Let~$\overline{\mathcal{M}}_{g,n}$ denote the Deligne--Mumford moduli 
space of stable algebraic curves of genus~$g$ with~$n$ distinct marked points~\cite{DM}. 
Let $\mathcal L_p$ be the $p$th tautological line bundle on~$\overline{\mathcal{M}}_{g,n}$,  and $\mathbb{E}_{g,n}$ the Hodge bundle.
Denote by $\psi_p:=c_1(\mathcal L_p)$, $p=1,\dots,n$, the first Chern class of $\mathcal L_p$, 
and by  $\lambda_j:= c_j(\mathbb{E}_{g,n})$, $j=0, \dots, g$, the $j$th Chern class of $\mathbb E_{g,n}$. 
The {\it Hodge integrals} are defined as the following {\it intersection numbers of mixed $\psi$-, $\lambda$-classes} 
on~$\overline{\mathcal{M}}_{g,n}$:
\beq\label{hodgekappaint}
\int_{\overline{\mathcal{M}}_{g,n}} \, \psi_1^{i_1}\cdots \psi_n^{i_n} \, \lambda_1^{j_1} \cdots \lambda_g^{j_g} \,,
\eeq
where $ i_1, \dots, i_n, j_1, \dots, j_g \geq 0$.
These integrals vanish unless
\beq\label{ddhodge}
(i_1\+\cdots\+ i_n) \+ (j_1\+2 \, j_2\+\cdots\+g \, j_g)\= 3 \, g \,-\, 3 \+ n \,. 
\eeq
For the case when all $j$'s are taken to be zero, the Hodge integrals~\eqref{hodgekappaint} are famously known as the 
intersection numbers of $\psi$-classes, which, according to Witten's conjecture~\cite{W} and Kontsevich's proof~\cite{K},
 have an explicit connection to the KdV hierarchy (cf.~e.g.~\cite{Dickey, DYZ, DZ-norm, K, W}). 

Let ${\rm ch}_k(\mathbb{E}_{g,n})$, $k\ge0$, be the components of the Chern character of~$\mathbb{E}_{g,n}$. 
According to Mumford~\cite{Mum}, the odd components of the Chern character of~$\mathbb{E}_{g,n}$ vanish. 
Denote by 
\beq\label{hodgeint}
\cH({\bf t};{\boldsymbol{\sigma}}; \e) \:= \sum_{g,n\geq0} \,\e^{2g-2}\, 
\sum_{i_1,\dots,i_n\geq0}\, \frac{t_{i_1}\cdots t_{i_n}}{n!} \, \int_{\overline{\mathcal{M}}_{g,n}}  
\psi_1^{i_1} \cdots \psi_n^{i_n} \cdot \exp\biggl(\sum_{j\ge1} \sigma_{2j-1} {\rm ch}_{2j-1}(\mathbb{E}_{g,n})\biggr)
\eeq
the generating series of Hodge integrals, called 
the {\it Hodge free energy}, 
and by~$\cH_g({\bf t}; {\boldsymbol{\sigma}})$ its genus~$g$ part, i.e.,
\beq
\cH({\bf t}; {\boldsymbol{\sigma}}; \e) \= \sum_{g\geq0} \, \e^{2g-2} \, \cH_g({\bf t}; {\boldsymbol{\sigma}})\,.
\eeq
Denote also 
\beq\label{hodgepart}
Z_{\rm H}({\bf t}; {\boldsymbol{\sigma}}; \e) \:= e^{\cH({\bf t}; {\boldsymbol{\sigma}}; \e)}.
\eeq
We call $Z_{\rm H}({\bf t}; {\boldsymbol{\sigma}}; \e)$ the {\it Hodge partition function}. The specialization  
\beq\label{wkpart}
Z_{\rm H}({\bf t}; {\bf 0}; \e)=:Z_{\rm WK} ({\bf t};\e), \quad  
\cH({\bf t}; {\bf 0}; \e) =: \F^{\rm WK} ({\bf t};\e) =: \sum_{g\ge0} \e^{2g-2} \F^{\rm WK}_g ({\bf t})
\eeq
are called the {\it Witten--Kontsevich partition function}, the {\it Witten--Kontsevich free energy}, respectively. 

We will be particularly interested in the Hodge integrals of the following form:
\beq\label{chbcy}
\int_{\overline{\mathcal{M}}_{g,n}} \, 
\Lambda_g(-1)^2 \, \Lambda_g\Bigl(\frac{1}{2}\Bigr) \, \psi_1^{i_1} \cdots \psi_n^{i_n} \,, 
\eeq
where $\Lambda_g(z) := \sum_{j=0}^g \lambda_j  z^j$ denotes the Chern polynomial of~$\mathbb E_{g,n}$.
Significance of the Hodge integrals in~\eqref{chbcy} 
was manifested by the Hodge-GUE correspondence~\cite{DLYZ2, DY1} (see also~\cite{DLYZ1}), by 
the Gopakumar--Mari\~no--Vafa conjecture on the string/Chern--Simons duality~\cite{GV, LLZ1,LLZ2,MV,OP}, 
by the Hodge-BGW correspondence~\cite{YZ22}, and etc. 
Generating series of these Hodge integrals can be obtained from~\eqref{hodgeint} by specializing the 
parameters $\sigma$'s as follows:
\beq
\sigma_{2j-1} \= - 2\, (1-4^{-j}) \, (2j-2)! , \quad j\geq1
\eeq
(see e.g.~\cite{DLYZ1}).  Denote 
\begin{align}
& \cH_{\rm special}({\bf t};\e) 
:= \cH \bigl(\bt; \{- 2\, (1-4^{-j}) \, (2j-2)! \} ; \e\bigr), \label{cubichodgeint} \\
& Z_{\rm special}({\bf t};\e) \:= Z_H\bigl({\bf t}; \{- 2\, (1-4^{-j}) \, (2j-2)! \};\e\bigr).
\end{align}
We call $\cH_{\rm special}({\bf t};\e)$ (and $Z_{\rm special}({\bf t};\e)$) 
the {\it Hodge free energy} (and respectively {\it Hodge partition function}) {\it associated to $\Lambda_g(-1)^2\Lambda_g(1/2)$}. 
We also denote by $\cH^{\rm special}_{g}({\bf t})$ the genus~$g$ part of the 
Hodge free energy associated to $\Lambda_g(-1)^2\Lambda_g(1/2)$, i.e.,
\beq
\cH_{\rm special}({\bf t};\e) \= \sum_{g\geq0} \, \e^{2g-2} \, \cH^{\rm special}_{g}({\bf t})\,.
\eeq

In~\cite{YZ22} an explicit relationship (see the following theorem) between the Hodge partition function associated to $\Lambda(-1)^2\Lambda(1/2)$ 
and the generalized BGW partition function, called the {\it Hodge-BGW correspondence}, was established. 
As in~\cite{YZ22}, define
\beq \label{def:gbgw-part}
\F(x, \bbT; \hbar) \:= \F_{\rm gBGW}\Bigl(\frac{x}{\hbar \,\sqrt{-2}}, \bbT; \hbar \Bigr) \+ B(x,\hbar)\,,
\eeq
where 
\beq \label{eqn:bernoulli}
B(x, \hbar) \= \frac1{\hbar^2} \biggl(\frac{x^2}{4} \, \log \Bigl(-\frac{x}{2}\Bigr) - \frac38 \, x^2\biggr) \+ \frac1{12} \, \log\Bigl(-\frac{x}{2}\Bigr)
\+ \sum_{g\geq 2} \frac{\hbar^{2g-2}}{x^{2g-2}}\frac{(-1)^g \, 2^{g-1} \, B_{2g}}{2g \, (2g-2)}
\eeq
with $B_k$ denoting the $k$th Bernoulli number. 
We call $\F(x, {\bf T}; \hbar)$  
the {\it generalized BGW free energy}, and its exponential 
\beq\label{def:Z}
\exp(\F(x, {\bf T}; \hbar)) =: Z(x, {\bf T}; \hbar)
\eeq
the {\it generalized BGW partition function}.
The genus expansion~\eqref{genusbgsalex} implies the genus expansion
\beq\label{eqn:genus-expan}
\F(x, {\bf T};\hbar) \;=:\; \sum_{g\ge0} \hbar^{2g-2}\, \F_g(x, {\bf T}),
\eeq
and we call $\F_g(x, {\bf T})$ the {\it genus $g$ part of the generalized BGW free energy}, for short the {\it genus $g$ generalized BGW free energy}.
We are ready to state the Hodge-BGW correspondence.
\begin{theorem}[Hodge-BGW correspondence~\cite{YZ22}]\label{HodgeBGWthm}
The following identity 
\begin{align}\label{mainid}
e^{\frac{A( x,\bbT)}{\hbar^2}} Z_{\rm special}\bigl({\bf t}( x, \bbT); \hbar\,\sqrt{-4}  \bigr) \=  Z( x, \bbT; \hbar) 
\end{align}
holds true in $\CC((\hbar^2))[[x+2]][[{\bf T}]]$.
Here,
\begin{align}\label{deftT}
&  t_i( x, \bbT) \= \delta_{i,0} \, x \+\delta_{i,1}\, -\, \Bigl(-\frac{1}{2}\Bigr)^{i-1}\,-\, 2\, \sum_{a\geq0} \,\Bigl(-\frac{2a+1}{2}\Bigr)^i  \;  \frac{ {T}_{2a+1} }{a!} \, , \quad i\geq 0\,,
\end{align}
and $A(x, \bbT)$ is a quadratic series given by
\beq\label{defA1215}
A(x,\bbT) \= \frac{1}{2} \, \sum_{a,b\geq0}
\, \frac{\widetilde { T}_{2a+1} \,\widetilde { T}_{2b+1}}{a! \, b! \, (a+b+1)} 
\, - \, \sum_{b\geq0} \, \frac{ x \, \widetilde { T}_{2b+1}}{b!\, (2b+1)}  \,,
\eeq
with $\widetilde{T}_{2a+1}  \=   T_{2a+1} \,-\, \delta_{a,0} $, $a\geq0$.
\end{theorem}

We note that Norbury~\cite{N} also gave a conjectural topological interpretation for $\log Z_{\rm BGW}$ 
(proved recently by Chidambaram, Garcia-Failde and Giacchetto~\cite{CGG}), and 
that Kazarian and Norbury~\cite{KN} gave another new topological interpretation 
(which is conjectured to be equivalent to the construction of~\cite{N}).

It is known that (cf.~\cite{DW, DLYZ1, DY2, EYY,GJV, IZ, ZZ}) 
for any $g\ge1$, the genus $g$ part of the Hodge free energy $\cH_g({\bf t}; {\boldsymbol{\sigma}})$ admits the  
jet-variable representation, i.e., for $g=1$, 
\beq
\cH_1({\bf t}; {\boldsymbol{\sigma}}) = H_1\biggl(v(\bt), \frac{\p v(\bt)}{\p t_0}; {\boldsymbol{\sigma}}\biggr), 
\quad H_1(z,z_1) := \frac1{24} \log z_1 +\frac{\sigma_1}{24}z \label{jetgenus1}
\eeq
and for $g\ge2$, there exists a unique function $H_g(z_1, \dots, z_{3g-2}; {\boldsymbol{\sigma}})$ of $(3g-2)$ variables, such that 
\beq\label{jetgenusg}
\cH_g({\bf t}; {\boldsymbol{\sigma}}) = H_g\biggl(\frac{\p v(\bt)}{\p t_0}, \dots, \frac{\p^{3g-2} v(\bt)}{\p t_0^{3g-2}}; {\boldsymbol{\sigma}}\biggr), \quad g\geq2.
\eeq
Here $v({\bf t})=\p_{t_0}^2 \F^{\rm WK}_0(\bf t)$.
For the special case when all ${\boldsymbol{\sigma}}={\bf 0}$, we know that  
\begin{align}
&\F^{\rm WK}_1({\bf t}) = F^{\rm WK}_1\biggl(\frac{\p v(\bt)}{\p t_0}\biggr),  \label{jstwk1-1}\\
&F^{\rm WK}_1(z_1) := \frac1{24} \log z_1\,, \label{jstwk1-2}
\end{align}
and that for $g\ge2$, there exists a unique function 
$F^{\rm WK}_g(z_1, \dots, z_{3g-2})$ of $(3g-2)$ variables, such that 
\beq\label{jstwk2}
\F^{\rm WK}_g({\bf t}) = 
F^{\rm WK}_g\biggl(\frac{\p v(\bt)}{\p t_0}, \dots, \frac{\p^{3g-2} v(\bt)}{\p t_0^{3g-2}}\biggr), \quad g\geq2.
\eeq
The unique functions $F^{\rm WK}_g(z_1, \dots, z_{3g-2})$ are the ones that appear in the Okuyama--Sakai conjecture. 

It follows from the jet-variable representation of $\cH^{\rm special}_{g}(\bt)$ (cf.~\eqref{jetgenus1}--\eqref{jetgenusg}) and Theorem~\ref{HodgeBGWthm} 
that the genus $g$ generalized BGW free energy $\F_g(x, {\bf T})$ for $g\ge1$ has the jet-variable representation: there exist functions 
$F_g(z_0, z_1,\dots,z_{3g-2})$, $g\ge1$, which for $g\ge2$ belong to $\CC[z_0^{\pm},z_1^{\pm1}][z_2,\dots,z_{3g-2}]$, such that 
\beq\label{jetfgx}
\F_g(x, {\bf T}) \= F_g \biggl( u(x,{\bf T}), \frac{\p u(x,{\bf T})}{\p x}, \dots, \frac{\p^{3g-2} u(x,{\bf T})}{\p x^{3g-2}} \biggr)\,, \quad g\ge1\,,
\eeq 
where 
\beq
u(x,{\bf T}) \:= - 4 \, \frac{\p^2 \F_0(x,{\bf T})}{\p x^2} \, . 
\eeq
(This representation was given in the Proposition~5 of~\cite{YZ22}). 
Recall from~\cite{YZ22} that the power series $u(x,{\bf T})$ is related to $v(\bt)$ by 
\beq\label{uvequal}
u(x,{\bf T}) \= v({\bf t}(x,{\bf T}))\,.
\eeq
Now let 
\beq\label{def:y}
y=y(x,{\bf T}) \:= \frac{\p^2 \F_0(x,{\bf T})}{\p T_1^2} 
\=\frac{\p^2 \F^{\rm gBGW}_0(x,{\bf T})}{\p T_1^2}\,.
\eeq
According to~\cite{YZ22} we know that 
\beq\label{relationsyu}
y\= e^{-u} \,, \quad 
u_x \= -\frac{y_x}y \,, \quad y_x \= - \frac{y_{T_1}}{2y^{1/2}} \,.
\eeq
Using~\eqref{relationsyu} and~\eqref{jetfgx}, we find that there exist functions $\bar F_g(z_0,z_1,\dots,z_{3g-2})$, $g\ge1$, such that
\beq\label{jetfgT1}
\F_{g}(x,{\bf T}) \= \bar F_g\bigg(y(x,{\bf T}), \frac{\p y(x,{\bf T})}{\p T_1},\dots,\frac{\p^{3g-2} y(x,{\bf T})}{\p T_1^{3g-2}}\bigg) \,,\quad g\ge1\,.
\eeq
These representations will be used in the next section. 

\section{Proof of the Okuyama--Sakai conjecture}\label{section3}
In this section we prove the Okuyama--Sakai conjecture by applying the Dubrovin--Zhang theory~\cite{DZ-norm} to the KdV hierarchy. 

Before proving Theorem~\ref{OSthm}, let us first recall the following well-known lemma.
\begin{lemma}\label{lem:eqns-wk}
For $g\geq 1$ the function $F^{\rm WK}_{g}=F^{\rm WK}_{g}(z_1,\dots,z_{3g-2})$ satisfies the following equations:
\begin{align}
\sum_{k\geq 1} \frac{k+2}{2} z_{k}\frac{\p F^{\rm WK}_{g}}{\p z_k}
\=&\frac{\delta_{g,1}}{16},\\
\sum_{k\geq 1}k z_k \frac{\p F^{\rm WK}_{g}}{\p z_k}
\=&(2g-2)F_{g}\+\frac{\delta_{g,1}}{24}.
\end{align}
\end{lemma}
\begin{proof}
It is well known that $Z_{\rm WK}({\bf t};\epsilon)$ satisfies the following two equations:
\begin{align}
& \sum_{i\geq 0} \frac{2i+1}{2}t_{i} \frac{\p Z_{\rm WK}({\bf t}; \e)}{\p t_i}\+ \frac{1}{16} Z_{\rm WK}({\bf t}; \e) \= \frac{3}{2}\frac{\p Z_{\rm WK}({\bf t}; \e)}{\p t_1} \, . \label{stringeq}\\
& \sum_{i\geq 0} t_i \, \frac{\p Z_{\rm WK}({\bf t}; \e)}{\p t_i} \+ \e \, \frac{\p Z_{\rm WK}}{\p \e} \+ \frac1{24} \, Z_{\rm WK}({\bf t}; \e) 
\= \frac{\p Z_{\rm WK}({\bf t}; \e)}{\p t_1} \,,  \label{dilatoneq}
\end{align}
The lemma is proved by using \eqref{wkpart}, \eqref{jstwk1-1}, \eqref{jstwk1-2}, \eqref{jstwk2}, \eqref{stringeq}, \eqref{dilatoneq}, and 
\beq
\frac{\p v({\bf t})}{\p t_i}\=\frac{v({\bf t})^i}{i!}\frac{\p v({\bf t})}{\p t_0}\,,\quad i\geq 0\,.
\eeq
\end{proof}
\begin{proof}[Proof of Theorem~\ref{OSthm}]
Let us first recall some properties of the genus $0$ part of the generalized BGW free energy.
Following~\cite{YZ22} (see (112) therein), introduce 
\beq
Q\=Q(x,{\bf T})\=\exp\Big({-\frac{u(x,{\bf T})}{2}}\Big)\,.
\eeq
Then $Q$ has~\cite{YZ22} the following properties:
\begin{align}
&\frac{\p Q}{\p T_{2a+1}} \= - \, \frac{2}{a!} \, Q^{2a+1} \, \frac{\p Q}{\p x} \,, \quad a\geq 0\,,\label{eqn:flow-Q}\\
&Q(x,{\bf 0}) \= - \, \frac{x}{2}\,.  \label{eqn:initial-Q}
\end{align}
Alternatively,  the power series $Q$ can be uniquely determined by the following equation~\cite{YZ22}:
	\beq\label{eqn:euler-lagrange-Q}
	Q \= - \, \frac{x}{2} \+ \sum_{a\geq0} \, {T}_{2a+1} \, \frac{Q^{2a+1}}{a!} \,.
	\eeq
Moreover, the genus $0$ free energy $\mathcal F_0(x,\bf T)$ of generalized BGW model satisfies the following equations~\cite{YZ22}:
\begin{align}
	\frac{\p^2 \F_0(x,\bbT) }{\p T_{2a+1} \p T_{2b+1} }
		& \= \frac{Q^{2a+2b+2}}{a!\, b!\,(a+b+1)}\,, \label{Q1} \\
	\frac{\p^2\F_0(x,\bbT)}{\p x \p T_{2b+1} } & 
\= - \, \frac{Q^{2b+1}}{b! \, (2 b+1)}\,, \label{Q2} \\
	\frac{\p^2\F_0(x,\bbT)}{\p x \p x} & \= \frac12\, \log Q \,. \label{Q3}
	\end{align}
Here $a,b\geq 0$. It has also been proved in~\cite{YZ22} that the genus zero part of the generalized BGW free energy $\F_0(x, \bbT)$ has the expression
\begin{align}\label{dubrovin-f0}
\F_0(x, \bbT) \= & \frac12 \, \sum_{a,b\geq0} \, \widetilde {T}_{2a+1} \, \widetilde {T}_{2b+1} \, \frac{Q^{2a+2b+2}}{a!\,b!\,(a+b+1)}
		\,-\, x \, \sum_{b\geq0} \, \widetilde {T}_{2b+1} \, \frac{Q^{2b+1}}{b!\,(2b+1)}  \+  \frac{x^2}{4} \, \log Q \,.
\end{align}

By using \eqref{def:y} and \eqref{Q1}, we find $y=Q^2$. 
Then by using \eqref{eqn:flow-Q} and \eqref{eqn:euler-lagrange-Q}, we have
\begin{align}
\frac{\p y}{\p T_{2a+1}}
&\=\frac{y^a}{a!}\,\frac{\p y}{\p T_1}\,,\quad a\geq 0\,,
\label{eqn:flow-v}
\end{align}
and
\beq
y|_{T_3=T_5=\cdots=0}\=\frac{x^2}{4\,(1-T_1)^2}\,.
\eeq
From~\eqref{Q1} (again noticing that $y=Q^2$), we know that $e^{\hbar^{-2}\F_{0}(x,{\bf T})}$ is the tau-function of the solution $y$ to the dispersionless KdV hierarchy~\eqref{eqn:flow-v}.

Define 
\beq
U(x, {\bf T};\hbar) \:= y(x, {\bf T}) \+ 
\sum_{g\geq 1} \hbar^{2g} \,\frac{\p^2 F_g^{\rm WK}\bigl(\frac{\p y(x, {\bf T})}{\p T_1},\dots,\frac{\p^{3g-2}y(x, {\bf T})}{\p T_1^{3g-2}}\bigr)}{\p T_1^2}.
\eeq
According to Dubrovin and Zhang~\cite{DZ-norm}, the power series~$U(x, {\bf T};\hbar)$ is a particular solution to the KdV hierarchy:
\beq\label{eqn:KdV}
\frac{\p U}{\p T_{2i+1}}\=\frac{1}{(2i+1)!!}\Big[\big(L^{\frac{2i+1}{2}}\big)_{+},U\Big],\quad i\geq0\,,
\eeq 
where $L=\hbar^2\p^2_{T_1}+2U$ is the Lax operator of the KdV hierarchy (cf. e.g.~\cite{Dickey}).
Moreover, the power series $\tau (x,{\bf T};\hbar)$, defined by
\beq\label{defitau10}
\tau (x,{\bf T};\hbar) \:=\exp\biggl(\hbar^{-2}\mathcal F_0(x, {\bf T})+\sum_{g\geq 1}\hbar^{2g-2}F_g^{\rm KW}\biggl(\frac{\p y(x, {\bf T})}{\p T_1},\dots,\frac{\p^{3g-2}y(x, {\bf T})}{\p T_1^{3g-2}}\biggr)\biggr) \,,
\eeq
is the tau-function of the solution $U(x, {\bf T};\hbar)$ to the KdV hierarchy~\cite{DZ-norm}. 
The conjectural identies~\eqref{eqn:bgwf1-wkjet}, \eqref{eqn:kw-bgwjet} are now equivalent to
\beq
Z(x,{\bf T};\hbar)\= 2^{-\frac{1}{24}} \, \tau(x,{\bf T};\hbar)\,.
\eeq

It is known (\cite{A,YZ22}) that $Z(x, {\bf T}; \hbar)$ satisfies the following equation:
\begin{align}\label{eqn:string-gBGW}
L_0 \bigl(Z(x, {\bf T}; \hbar)\bigr) \= 0\,,
\end{align}
where
\begin{align}
L_0\=&\sum_{a\geq 0}\frac{2a+1}{2}\widetilde T_{2a+1}\frac{\p}{\p T_{2a+1}}
\+\frac{1}{16}+\frac{x^2}{8\hbar^2}\,.
\end{align}

By equation~\eqref{eqn:string-gBGW} and \eqref{def:Z}, \eqref{eqn:genus-expan}, $\mathcal F_{0}(x,{\bf T})$ satisfies the following equation
\beq\label{eqn:string-gBGW-F0}
\sum_{a\geq 0}\frac{2a+1}{2}\widetilde{T}_{2a+1}\frac{\p \F_{0}(x,{\bf T})}{\p T_{2a+1}}
\+\frac{x^2}{8}\=0.
\eeq
By repeatedly taking derivatives of the above equation with respect to $T_1$
and by using the commutation relation
\beq
\Big[\frac{\p}{\p T_1},L_0\Big]\=\frac{1}{2}\,\frac{\p}{\p T_1}\,,
\eeq
we get
\beq
\sum_{a\geq 0}\frac{2a+1}{2}\widetilde{T}_{2a+1}\frac{\p }{\p T_{2a+1}}\bigg(\frac{\p^k y}{\p T_1^k}\bigg)\=-\frac{k+2}{2}\frac{\p^k y}{\p T_1^k}\,,\quad k\geq 0\,.
\eeq
Therefore, 
\beq
\sum_{a\geq 0}\frac{2a+1}{2}\widetilde{T}_{2a+1}\frac{\p F_g^{\rm KW}\bigl(\frac{\p y}{\p T_1},\dots,\frac{\p^{3g-2}y}{\p T_1^{3g-2}}\bigr)}{\p T_{2a+1}}
\= \Biggl(-\sum_{k=1}^{3g-2}\frac{k+2}{2} z_k\frac{\p F_g^{\rm WK}}{\p z_k}\Biggr)\biggl(\frac{\p y}{\p T_1},\dots,\frac{\p^{3g-2}y}{\p T_1^{3g-2}}\biggr)\,.
\eeq
Together with Lemma~\ref{lem:eqns-wk}, we arrive at
\beq\label{eqn:string-gBGW-Fg}
\sum_{a\geq 0}\frac{2a+1}{2}\widetilde{T}_{2a+1}\frac{\p F_g^{\rm KW}\bigl(\frac{\p y}{\p T_1},\dots,\frac{\p^{3g-2}y}{\p T_1^{3g-2}}\bigr)}{\p T_{2a+1}}
\+\frac{\delta_{g,1}}{16}
\=0\,\quad g\geq 1\,.
\eeq
Hence
\begin{align}
L_0 \big(\tau(x,{\bf T};\hbar)\big)\=0\,.\label{eqn:string-gBGW-jet}
\end{align}

It is also known (\cite{A,YZ22}) that $Z(x, {\bf T}; \hbar)$ satisfies the following dilaton equation:
\begin{align}\label{eqn:dilaton-gBGW}
L_{\rm dilaton}\bigl(Z(x, {\bf T}; \hbar)\bigr) \= 0\,,
\end{align}
where
\begin{align}
L_{\rm dilaton} \= &\sum_{a\geq 0}\,\widetilde T_{2a+1}\,\frac{\p}{\p T_{2a+1}}\+x\,\frac{\p}{\p x}\+\hbar\,\frac{\p }{\p \hbar}\+\frac{1}{24} \,.
\end{align}

By equation~\eqref{eqn:dilaton-gBGW} and \eqref{def:Z}, \eqref{eqn:genus-expan}, $\mathcal F_{0}(x,{\bf T})$ satisfies the following equation
\beq\label{eqn:dilaton-gBGW-F0}
\sum_{a\geq 0}\widetilde{T}_{2a+1}\frac{\p \F_{0}(x,{\bf T})}{\p T_{2a+1}}
\+x\,\frac{\p \F_{0}(x,{\bf T})}{\p x}\=2\, \F_{0}(x,{\bf T}).
\eeq
Like the above, using the commutation relation
\beq
\Big[\frac{\p}{\p T_1},L_{\rm dilaton}\Big]\=\frac{\p}{\p T_1}\,,
\eeq
we get
\beq
\sum_{a\geq 0}\widetilde{T}_{2a+1}\frac{\p }{\p T_{2a+1}}\bigg(\frac{\p^k y}{\p T_1^k}\bigg)
\+x\frac{\p }{\p x}\bigg(\frac{\p^k y}{\p T_1^k}\bigg)
\=-k\,\frac{\p^k y}{\p T_1^k}\,,\quad k\geq 0\,.
\eeq
Together with Lemma~\ref{lem:eqns-wk}, we obtain 
\begin{align}
L_{\rm dilaton}\big(\tau(x,{\bf T};\hbar)\big)\=0\,.\label{eqn:dilaton-gBGW-jet}
\end{align}

Denote
\beq\label{defiwtf310}
\log(\tau (x,{\bf T};\hbar))
\,=:\, \widetilde{\mathcal F}(x,{\bf T};\hbar)
\,=:\, \sum_{g\geq 0}\hbar^{2g-2}\, \widetilde{\mathcal F}_g(x,{\bf T}) \, .
\eeq
By equation~\eqref{eqn:string-gBGW} and equation~\eqref{eqn:string-gBGW-jet}, set $T_{2a+1}=0$, $a=1,2,\cdots$,
we have
\beq\label{eqn:tildeF-pT}
\hbar^2\frac{\p \widetilde{\mathcal F}}{\p T_1}(x,T_1, 0, 0, \dots;\hbar)
\=\hbar^2\frac{\p {\mathcal F}}{\p T_1}(x,T_1, 0, 0, \dots;\hbar)
\=\frac{1}{1-T_1} \biggl(\frac{\hbar^2}{8}+\frac{x^2}{4}\biggr)\,.
\eeq
Taking the derivative with respective to $T_1$ in the above equation, we get
\beq
U(x,T_1, 0, 0, \dots;\hbar)
\=U_{\rm gBGW}(x,T_1, 0, 0, \dots;\hbar)
\=\frac{1}{(1-T_1)^2} \biggl(\frac{\hbar^2}{8}+\frac{x^2}{4}\biggr)\,,
\eeq
where $U_{\rm gBGW}(x,{\bf T};\hbar):=\hbar^2\frac{\p^2 {\mathcal F}}{\p T_1^2}(x,{\bf T};\hbar)$.
So we have proved that these two power series, $U(x,{\bf T};\hbar)$ and $U_{\rm gBGW}(x,{\bf T};\hbar)$, 
both satisfy the KdV hierarchy and have the same initial value,
thus by uniqueness of the solution to the KdV hierarchy, we have 
\beq
U(x,{\bf T};\hbar)=U_{\rm gBGW}(x,{\bf T};\hbar).
\eeq

It follows that $\F(x,{\bf T}; \e)$ and $\widetilde{\mathcal F}(x,{\bf T};\e)$ 
could only differ by an affine function of $T_1, T_2, \dots$ (the coefficients can depend on $x$ and~$\e$). 
From~\eqref{eqn:string-gBGW} and~\eqref{eqn:string-gBGW-jet} we know that 
this affine function could only be a constant with respect to $T_1, T_2, \dots$, that is a function of $x$, $\e$. 
Combining with the genus expansions~\eqref{eqn:genus-expan} and~\eqref{defiwtf310}, we can now write 
\beq
\F(x,{\bf T}; \e) - \widetilde{\mathcal F}(x,{\bf T};\e) \, = : \, \sum_{g\geq0} \e^{2g-2} \mathcal{K}_g(x) \,.
\eeq
We are left to show that the functions $\mathcal{K}_g(x)$, $g\ge0$, all vanish. 
Using \eqref{defitau10}, \eqref{defiwtf310}, \eqref{jetgenus1}, \eqref{uvequal}, \eqref{eqn:genus-expan}, \eqref{jstwk1-2}, we find that this is true for $g=0,1$. 
Using equation~\eqref{eqn:dilaton-gBGW} and equation~\eqref{eqn:dilaton-gBGW-jet}, 
we know that $\mathcal{K}_g(x)$ for $g\ge2$ must have the form 
\beq
\mathcal{K}_g(x) \= \frac{c_g}{x^{2g-2}}\,, \quad g\ge2\,.
\eeq
Because of \eqref{jetfgT1} and \eqref{defitau10}, \eqref{defiwtf310}, we know that for each $g\ge2$ 
there exists a function $K_g(z_0, z_1,\dots,z_{3g-2})\in\CC[z_0^{\pm1},z_1^{\pm1}][z_2,\dots,z_{3g-2}]$ such that 
\beq\label{kkequal}
\mathcal{K}_g(x) \=  K_g\biggl(y(x,{\bf T}), \frac{\p y(x, {\bf T})}{\p T_1},\dots,\frac{\p^{3g-2}y(x, {\bf T})}{\p T_1^{3g-2}}\biggr) \,,
\eeq 
where $y(x,{\bf T})$ is given by~\eqref{def:y}.
Like in the first arXiv preprint version of~\cite{DLYZ2}, 
taking derivatives with respect to $T_{2m+1}$, $m\ge0$, on both sides of~\eqref{kkequal}, using~\eqref{eqn:flow-v},
and dividing both sides by $y(x,{\bf T})^m$, we find that the right-hand side becomes a polynomial of~$m$, 
which vanishes {\it identically} in~$m$. 
Comparing the coefficients of powers of~$m$ from the highest degree to the lowest degree we obtain 
\beq\label{vanishingfirstderiv}
\frac{\p K_g}{\p z_k}\biggl(y(x,{\bf T}), \frac{\p y(x, {\bf T})}{\p T_1},\dots,\frac{\p^{3g-2}y(x, {\bf T})}{\p T_1^{3g-2}}\biggr) \= 0 \,, \quad 
k=0,\dots,3g-2\,.
\eeq
This, by an elementary exercise, leads to the conclusion that for each $g\ge2$ the Laurent polynomial 
$K_g(z_0, z_1,\dots,z_{3g-2})$ must be a pure constant. Therefore, $c_g=0$.
The theorem is proved.
\end{proof}

We note that it is shown in~\cite{YZ22} that 
Theorem~\ref{OSthm} implies Kazarian--Norbury's conjectural identity~\cite{KN} for kappa class integrals.

\medskip

\noindent {\it Note added:}  
After the first version of this paper appeared on arXiv,
 Chekhov kindly communicated to us the paper~\cite{AC}, where 
 the so-called Born--Infeld (NBI) matrix model was considered which is very 
 similar to the generalized BGW model. 
 It is shown in~\cite{AC} that the partition function of the NBI model can be identified with the Witten--Kontsevich 
 tau-function by shifting times (by constants) in the power-series ring, therefore with 
 the partition function of certain kappa class integrals (see e.g.~\cite{BDY} and the references therein). However, the corrected/normalized 
 partition function of the generalized BGW model is a power series of times 
 which cannot be obtained by shifting times by constants from the Witten--Kontsevich 
 tau-function. A simple way to see this is that the Virasoro constraints for the NBI model start from~$L_{-1}$ while 
 for the generalized BGW model they start from~$L_0$. Nevertheless, it seems to us that the tau-function 
 given from the viewpoint of~\cite{DYZ} for the KdV hierarchy can unify the two models (see the right part of (115) in~\cite{DYZ}). 
 We study the precise relation of these two models in subsequent publications.

\medskip
\medskip

\noindent Di Yang

\noindent School of Mathematical Sciences, University of Science and Technology of China,

\noindent Hefei 230026, P.R. China 

\noindent diyang@ustc.edu.cn

\medskip

\noindent Qingsheng Zhang

\noindent School of Mathematical Sciences, Peking University, 

\noindent Beijing 100871, P.R. China 

\noindent zqs@math.pku.edu.cn

\end{document}